%% file: densities_monoid.tex
\documentclass[a4paper,UKenglish,numberwithinsect,cleveref]{lipics-v2021}
\hideLIPIcs

\usepackage{tikz}
\usetikzlibrary{positioning,matrix,fit,backgrounds}

\tikzset{>=latex}
\tikzstyle{every loop}=[->,shorten >=1pt,looseness=8]
\tikzstyle{loop left}=[in=210,out=150,loop,left]
\tikzstyle{loop right}=[in=-30,out=30,loop,right]
\tikzstyle{state}=[circle,draw,inner sep=2pt,minimum size=.4cm]
\tikzstyle{label}=[auto,font=\small]
\tikzstyle{idempotent}=[font=\scriptsize,anchor=north east,inner sep=1pt]
\tikzstyle{dclass}=[matrix of nodes,draw,inner sep=0pt,nodes={font=\small},minimum width=20pt,minimum height=20pt]
\newcommand{\drawgrid}[3]{
            \foreach \i in {2,...,#1}
                \draw (#3-\i-1.north west) to (#3-\i-#1.north east);
            \foreach \i in {2,...,#2}
                \draw (#3-1-\i.north west) to (#3-#2-\i.south west);
            }

\newcommand{\HH}{\mathrel{\mathscr{H}}}
\newcommand{\JJ}{\mathrel{\mathscr{J}}}
\newcommand{\DD}{\mathrel{\mathscr{D}}}
\newcommand{\RR}{\mathrel{\mathscr{R}}}
\newcommand{\LL}{\mathrel{\mathscr{L}}}
\newcommand{\A}{\mathcal A}

\newcommand{\Z}{\mathbb{Z}}
\newcommand{\N}{\mathbb{N}}
\newcommand{\R}{\mathbb{R}}
\newcommand{\Q}{\mathbb{Q}}
\newcommand{\K}{\mathbb{K}}
\newcommand{\cL}{\mathcal L}

\newcommand*{\from}{\colon}
\newcommand*{\Card}{\operatorname{Card}}
\newcommand*{\Image}{\operatorname{im}}
\newcommand*{\rank}{\operatorname{rk}}

\title{Density of rational languages under shift invariant measures}

\author{Valérie Berthé}{Université Paris Cité, IRIF, CNRS}{berthe@irif.fr}{https://orcid.org/0000-0001-5561-7882}{}
\author{Herman Goulet-Ouellet}{Faculty of Information Technology, Czech Technical University in Prague}{herman.goulet.ouellet@fit.cvut.cz}{https://orcid.org/0000-0003-3445-8469}{The second author was supported by the CTU Global Postdoc Fellowship program.}
\author{Dominique Perrin}{LIGM, Université Gustave Eiffel}{dominique.perrin@esiee.fr}{https://orcid.org/0000-0003-4036-141X}{}
\authorrunning{V. Berthé, H. Goulet-Ouellet and D. Perrin}

\Copyright{Valérie Berthé, Herman Goulet-Ouellet and Dominique Perrin} 

\category{Track B: Automata, Logic, Semantics, and Theory of Programming}

\nolinenumbers 

\EventEditors{Keren Censor-Hillel, Fabrizio Grandoni, Joel Ouaknine, and Gabriele Puppis}
\EventNoEds{4}
\EventLongTitle{52nd International Colloquium on Automata, Languages, and Programming (ICALP 2025)}
\EventShortTitle{ICALP 2025}
\EventAcronym{ICALP}
\EventYear{2025}
\EventDate{July 8--11, 2025}
\EventLocation{Aarhus, Denmark}
\EventLogo{}
\SeriesVolume{334}
\ArticleNo{162}

\begin{document}

\maketitle

\begin{abstract}
    We study density of rational languages under shift invariant probability measures on spaces of two-sided infinite words, which generalizes the classical notion of density studied in formal languages and automata theory.  The density for a language  is  defined as the limit in average (if it exists) of the probability that a word of a given length  belongs to the language. We establish the existence of densities for all rational languages under all shift invariant measures. We also give explicit formulas under certain conditions, in particular when the language is aperiodic. Our approach combines tools and ideas from semigroup theory and ergodic theory.
\end{abstract}

The natural density for a language $L$ on a finite alphabet $A$ is  defined as the limit
in average (if it exists) of the probability that a word of length $n$ belongs to $L$,
where letters are drawn independently with equal probabilities. This notion can be traced
back to the work of Berstel~\cite{Berstel1972}, who took inspiration from
Schützenberger~\cite{Schutzenberger1965}. It has been widely studied in the context of automata theory, logic and the theory of codes~\cite{BerstelPerrinReutenauer2009,book/Salomaa1978,Lynch1993,Hansel1983,Hansel1989,Bodirsky2004,Sinya2015,Koga2019,Kozik2005}, and also appears in ergodic theory, for instance in the work of Veech~\cite{Veech1969,Veech1975}. Densities have also been studied in the more general case where words are drawn with a Bernoulli measure, i.e.\ letters are drawn independently with possibly different probabilities~\cite{BerstelPerrinReutenauer2009}.

This paper is related with recent work on the density of group languages~\cite{BertheGouletOuelletNybergBroddaPerrinPetersen2024} using a more general notion of density motivated by symbolic dynamics. Let $\mu$ be a probability measure on the space $A^\Z$ of two-sided infinite words on the alphabet $A$. Then the density of a language $L$ with respect to $\mu$ is the limit
\begin{equation}
    \delta_\mu(L)=\lim_{n\to\infty}\frac{1}{n}\sum_{i=0}^{n-1}\mu(\{x\in A^\Z \mid x_0\dots x_{i-1}\in L\}),\label{eqDefinitionDensity}
\end{equation}
if it exists. The classical definition  discussed above corresponds to the case where $\mu$ is a Bernoulli measure. Our main result states that the density of a rational language $L$ under a shift invariant probability measure $\mu$ always exists (\cref{theoremMain}). A shift invariant measure  is a probability measure that behaves well with respect to left extensions of words (see Equation~\eqref{eq:inv}). The proof closely combines dynamical and algebraic methods. 
 
 On the dynamical side, the proof relies  mainly on  the construction of a skew product between the shift space (made of two-sided infinite words)  that supports  the measure
 $\mu$ and an $\RR$-class of the finite monoid $M$  defining  the rational language  $L$ (the transition monoid). The skew product we use is closely related to the notion of wreath product used in semigroup and automata theory~\cite[Chapter 1, Section 10]{Eilenberg1976}. We  construct a natural measure on the skew product, called the weighted counting measure, which leads to a closed formula for the density under the condition that this measure is ergodic (\cref{theoremWeightedCountingMeasure}). This generalizes known results for the case of Bernoulli measures. We also consider in \cref{theoremDensityAperiodic}   aperiodic languages (also called star-free)  for which densities are proved to exist in a  strong sense.

On the algebraic side, this construction relies crucially on the theory of Green's relations and on the key notion of the $\JJ$-class associated with a shift. Let the monoid $M$ be the image of $A^*$ by a morphism $\varphi$. The $\JJ$-class of $M$ associated with $X$ is the set of generators of the least ideal of $M$ which meets the image by $\varphi$ of the language of the shift  $X$.
Therefore, it can be called the minimal $\JJ$-class of the monoid $M$ with respect to the shift $X$.
The use of such a $\JJ$-class is a useful tool in automata theory, and appears in several
different contexts. See \cite{Colcombet2011} for a survey on the use of this idea,
originating in the proof by Sch\"utzenberger of the characterization of aperiodic languages
and \cite{PerrinSchupp1986} for an application close to the present paper.

Let us give a brief outline of the paper. In \cref{sec:dynamics} we present preliminaries on symbolic dynamics. In \cref{sec:idealsfullmonoid} we show how to calculate the densities of ideals of $A^*$. In \cref{sec:Jclass}, we study the $\JJ$-class associated with a shift space in a finite monoid and derive an explicit formula for the density of aperiodic languages (\cref{theoremDensityAperiodic}). We also  discuss connections with the notion of degree of an automaton and decidability results concerning the $\JJ$-class. In \cref{sec:existence}, we introduce the skew products and weighted counting measures, give the proof of our main result (\cref{theoremMain}), along with an explicit expression for the density when the weighted counting measure is ergodic (\cref{theoremWeightedCountingMeasure}). In \cref{sec:algebraic}, we consider algebraic properties of the density (see \cref{corollaryExtensionField}). Lastly, \cref{sec:markov} contains a discussion on the case of sofic measures, motivated by  \cref{corollaryExtensionField}.
Additionally, we provide in \cref{appendix-semigroups} a brief introduction to semigroup theory, focusing on ideals and Green's relations.

\section{Symbolic dynamics}
\label{sec:dynamics}

This section covers some necessary material from symbolic dynamics. We refer to \cite{DurandPerrin2021} for a complete exposition on the topic, including proofs, and also to \cite{Fog02,Queffelec:10} as alternative sources. For more specialized texts on topological dynamics and ergodic theory, we refer to \cite{book/Petersen1983,Walters1982}.

\subsection{Topological dynamical systems}
\label{subsec:top-dynamics}

We first recall some terminology from topological dynamics. A \emph{topological dynamical system} 
is a pair $(X,T)$ of a compact metric space $X$ and a continuous transformation $T\colon X\to X$. A subset $Y\subseteq X$ is called \emph{stable} if $T(Y)\subseteq Y$. A nonempty topological dynamical system $(X,T)$ is \emph{minimal} if the only closed stable subsets of $X$ are $X$ and $\varnothing$. Equivalently, all $x\in X$ have dense forward orbits $\{T^n(x)\mid n\ge 0\}$. As a weaker notion, $(X,T)$ is \emph{transitive} if there exists $x\in X$ with dense forward orbit. Transitivity is equivalent to the condition that $U\cap T^{-n}V\ne\emptyset$ for all pairs of open sets $U,V\subseteq X$.

Let $\mu$ be a Borel probability measure on $X$. The \emph{support} of $\mu$ is the smallest closed set of measure 1. We say that $\mu$ is \emph{invariant} if $\mu(T^{-1}U)=\mu(U)$ for every Borel set $U\subseteq X$. Note that the support of an invariant measure is a closed stable subset. 

The measure $\mu$ is called \emph{ergodic} if it is invariant and every Borel set such that $T^{-1}(U)=U$ has measure $0$ or $1$. By Birkhoff's ergodic theorem, ergodicity can be interpreted as a form of asymptotic independence. More precisely, an invariant measure $\mu$ is ergodic if and only if
\begin{equation}
    \lim_{n\to\infty}\frac{1}{n}\sum_{i=0}^{n-1}\mu(U\cap T^{-i}V)=\mu(U)\mu(V)
    \label{eqErgodic}
\end{equation}
for every pair $U,V$ of Borel sets (see \cite{book/Petersen1983}). In particular, note that the support of an ergodic measure is a transitive system. As a stronger property, the measure $\mu$ is \emph{mixing} if 
  \begin{equation}
    \lim_{n\to\infty}\mu(U\cap T^{-n}V)=\mu(U)\mu(V)
    \label{eqMixing}
\end{equation}
for every pair $U,V$ of Borel sets. Every mixing measure is ergodic but the converse is false. 

Let $\mathcal{M}=\mathcal{M}(X,T)$ be the set of invariant probability measures on $(X,T)$ and $\mathcal{E} = \mathcal{E}(X,T)$ be the subset of those measures which are ergodic. We can view $\mathcal{M}$ as a subset of the dual $\mathcal C(X,\R)^*$ of the space of continuous maps $X\to \R$. If we equip $\mathcal C(X,\R)^*$ with the weak-$*$ topology, then $\mathcal{M}$ is a compact convex subset and $\mathcal{E}$ is its set of extreme points. In particular if the system $(X,T)$ has only one invariant measure, it must be ergodic; and we call $(X,T)$ \emph{uniquely ergodic}. Standard results from functional analysis imply that any invariant measure can be decomposed in terms of ergodic measures. More precisely, for every invariant measure $\mu\in\mathcal{M}$, there exists a Borel probability measure $\tau$ on $\mathcal E$ such that
\begin{equation}
    \mu=\int_{\mathcal E}\nu \,d\tau(\nu).\label{eqErgodicDecomposition}
\end{equation}
More details can be found in e.g. \cite[Chapter 12]{book/Phelps2001}. 

\subsection{Shift spaces}\label{subsec:defsubshift}

We now turn to symbolic dynamics. Let $A$ be a finite alphabet. We denote by $A^\Z$ the set of two-sided infinite sequences over $A$ equipped with the product topology, where $A$ has the discrete topology. We denote by $S$ the shift transformation on $A^\Z$, defined by $y=S(x)$ if $y_n=x_{n+1}$ for all $n\in\Z$. A \emph{shift space} $X$ on the alphabet $A$ is, by definition, a closed subset of $A^\Z$ which is stable under the shift transformation. Note that $(X,S)$ is a topological dynamical system. 

Let $A^*$ be the free monoid on $A$, $\varepsilon$ be the empty word, and $A^+=A^*\setminus\{\varepsilon\}$. We denote by $\cL(X)$ the set of finite words $w$ which appear in the elements $x\in X$ and by $\cL_n(X)$ the set of words of length $n$ in $\cL(X)$. For a shift space, transitivity and minimality can be characterized in terms of the language $\cL(X)$. A shift space $X$ is transitive if and only if for every $u,v\in\cL(X)$, there is some $w\in\cL(X)$ such that $uwv\in\cL(X)$. Moreover $X$ is minimal if and only if for every $n\ge 0$, there is an $N\ge 0$ such that every word of $\cL_n(X)$ appears in every word of $\cL_N(X)$. Shift spaces which are transitive are also called \emph{irreducible}.

For words $u,v$ of length $m,n$ respectively, we define the right and two-sided cylinders by 
\begin{equation*}
    [v]_X=\{x\in X\mid x_0\dots x_{n-1}=v\},\qquad [u\cdot v]_X=\{x\in X\mid x_{-m}\dots x_{n-1}=uv\}.
\end{equation*}
The two-sided cylinders form a clopen basis for the topology of $X$. For $L,K\subseteq A^*$, we define 
\[
    [L]_X=\bigcup_{w\in L}[w]_X,\qquad [L\cdot K]_X=\bigcup_{u\in L,v\in K}[u\cdot v]_X.
\]

Let $\mu$ be a Borel probability measure on a shift space $X$. By a slight abuse of notation, we also denote by $\mu$ the map which assigns to a word $w\in \cL(X)$ the number $\mu([w]_X)$. With this notation, we have $\mu(\varepsilon)=1$ and 
\begin{equation*}
    \sum_{a\in A}\mu(wa)=\mu(w).
\end{equation*}
We extend this to subsets $L\subseteq\cL(X)$ by letting $\mu(L)=\sum_{w\in L}\mu(w)$. 
When the cylinders $[u]_X$ for $u\in L$ are disjoint, we have $\mu(L)=\mu([L]_X)$. 
Moreover, if $\mu$ is an invariant measure, then 
\begin{equation}\label{eq:inv}
    \mu(w) = \mu([w\cdot\varepsilon]_X) = \sum_{a\in A}\mu(aw).
\end{equation}
When the cylinders $[u\cdot\varepsilon]_X$ for $u\in L$ are disjoint, we have $\mu(L)=\mu([L\cdot\varepsilon]_X)$.

Observe that the support of an invariant measure on $A^\Z$ is a shift space $X$ of measure 1, and when the measure is ergodic, $X$ is irreducible.

A \emph{Bernoulli measure} is the simplest case of an ergodic measure on $A^\Z$. The values $\mu(w)$ for $w\in A^*$ are given by a morphism $\mu\colon A^*\to[0,1]$ such that $\sum_{a\in A}\mu(a)=1$. It corresponds, in classical terms of probability theory, to a sequence $(\zeta_n)_{n\in\N}$ of independent identically distributed random variables, where $\zeta_n\from x\mapsto x_n$. A Bernoulli measure is in fact mixing. Further mixing examples can be found among sofic measures, also called rational measures or hidden Markov measures~\cite{Hansel1989,BoylePetersen2011}. 
See \cref{sec:markov} for a discussion on this topic.

Another important source of examples is the following. A \emph{substitution} is a monoid morphism $\sigma\colon A^*\to A^*$. The substitution $\sigma$ is \emph{primitive} if there is an $n\ge 0$ such that every $b\in A$ occurs in every $\sigma^n(a)$. The shift $X(\sigma)$, called a \emph{substitution shift}, is the set of all $x\in A^\Z$ such that all words in $\cL(x)$ appear in some $\sigma^n(a)$, $n\ge 0$, $a\in A$. If $\sigma$ is primitive then $X(\sigma)$ is minimal, and uniquely ergodic~\cite{Michel1974}. 

\begin{example}\label{exampleFibonacci}
    The Fibonacci substitution is the morphism $\sigma\colon \{a,b\}^*\to \{a,b\}^*$ defined by $\sigma\colon a\mapsto ab,b\mapsto a$. It is primitive and the substitution shift $X=X(\sigma)$ is called the \emph{Fibonacci shift}. A description of its unique ergodic measure can be found e.g.\ in \cite[Example~3.8.19]{DurandPerrin2021}.
\end{example}

The Fibonacci shift also belongs to the family of \emph{Sturmian shifts}, whose definition may be recalled in~\cite{Lothaire1983}. More precisely, the Fibonacci shift is Sturmian of slope $(3-\sqrt{5})/2$. The following is an example of an automatic sequence (see~\cite{book/Allouche2003}).

\begin{example}\label{exampleMorse}
    The primitive substitution $\sigma\colon a\mapsto ab,b\mapsto ba$ is called the \emph{Thue--Morse substitution}. The shift $X(\sigma)$ is called the \emph{Thue--Morse shift}. Its unique ergodic measure is described in \cite[Example 3.8.20]{DurandPerrin2021}.
\end{example}

\subsection{Density of a language}

Let $L$ be a  language on $A$ and let $\mu$ be an ergodic measure on $A^\Z$ with support $X$. We define the \emph{density} of $L$ relative to $\mu$ as
\begin{displaymath}
    \delta_\mu(L)=\lim_{n\to\infty}\frac{1}{n}\sum_{i=0}^{n-1}\mu(L\cap A^i),
\end{displaymath}
whenever the limit exists. In other words $\delta_\mu(L)$ is the Cesàro limit of $(\mu(L\cap A^n))_{n=0}^\infty$. When the classical limit of $(\mu(L\cap A^n))_{n=0}^\infty$ exists, we say that $L$ has a density \emph{in the strong sense}. It is known that when $\mu$ is a Bernoulli measure, the density of any rational language exists, and if $\mu$ has rational values on letters, all densities are rational numbers~\cite[Theorem~2.1]{Berstel1972}.

Let us note some basic properties of the density. If $L$ has a density, then  $0\le \delta_\mu(L)\le 1$, as $0\le\mu(L\cap A^i)\le 1$ for every $i\ge 0$. Moreover, since $\mu(w)=0$ when $w\notin\cL(X)$, we have $\delta_\mu(L)=\delta_\mu(L\cap\cL(X))$. Thus if $L,K$ satisfy $L\cap\cL(X)=K\cap\cL(X)$ then $\delta_\mu(L)=\delta_\mu(K)$. The density is also finitely additive: if $L,K$ have densities and $L\cap K=\varnothing$, then $L\cup K$ has density $\delta_\mu(L\cup K)=\delta_\mu(L)+\delta_\mu(K)$. Additionally, we have $\delta_\mu(A^*\setminus L)=1-\delta_\mu(L)$, and if $L\cup L'=A^*$ then $L\cap L'$ has density $\delta_\mu(L\cap L') = \delta_\mu(L)+\delta_\mu(L')-1$. However, the density of an intersection $L\cap K$ might not exist, even when the density of both $L$ and $K$ exist.
\begin{example}
    On a fixed finite alphabet $A$, consider the two languages:
    \begin{equation*}
        L = \{ w\in A^* : |w| \equiv 1 \mod 2 \},\qquad K = \{ w\in A^+ : |w|\equiv \lfloor\log_2|w|\rfloor \mod 2 \}.
    \end{equation*}
    It is clear that $\delta_\mu(L)=\delta_\mu(K)=1/2$ (no matter the measure $\mu$). However, it can be shown that the sequence of sums $s_n = \frac1{n}\sum_{i=1}^{n}\mu(L\cap K\cap A^i)$ has subsequences $(s_{2^{2n}})_{n\in\N}$ and $(s_{2^{2n+1}})_{n\in\N}$ converging to $1/3$ and $1/6$ respectively, thus $\delta_\mu(L\cap K)$ does not exist.
\end{example}

In the above example, the language $K$ is not rational. When both languages are rational our main result (\cref{theoremMain}) implies that the density of the intersection also exists.

\section{Density of ideals} 
\label{sec:idealsfullmonoid}

In this section, we show how to calculate densities for languages that are  ideals of $A^*$ (and thus not necessarily rational). We start with {\em right ideals}, that is, languages $L$ such that $LA^*=L$. The set $D=L\setminus L A^+$ is the  minimal generating set of $L$ as a right ideal, which means that $L= DA^*$ and $D$ is contained in every other set with this property. Moreover the set $D$ is a \emph{prefix code}, meaning that no element of $D$ is a proper prefix of another element of $D$. For example, for $a\in A$, the language $L= aA^*$ is the set of words  that have the letter $a$ as a prefix. Then one has $D=\{a\}$.

\begin{proposition}\label{propositionRightIdeal}
    Let $\mu$ be a probability measure on $A^\Z$ and $L$ be a right ideal of $A^*$. Then
    $L$ has a density in the strong sense and $\delta_\mu(L)=\mu(D)$ where $D = L\setminus LA^+$.
\end{proposition}

\begin{proof}
    Let $X$ be the support of $\mu$. Since $D$ is a prefix code,
    \begin{equation*}
            \mu(DA^*\cap A^n)=\sum_{u\in D}\mu(uA^*\cap A^n)=\sum_{u\in D, |u|\le n}\mu(uA^{n-|u|})=\sum_{u\in D,|u|\le n}\mu(u)
    \end{equation*}
    where the last equality uses the fact that $\mu(uA^i) = \mu(u)$ for all $i$. Clearly this tends to $\mu(D)$ when $n\to\infty$.
\end{proof}

Next we establish a similar result for {\em left ideals} of $A^*$, that is, languages $L$ such that $A^*L=L$. Similar to the right-sided case, the set $G = L\setminus A^+L$ is the minimal generating set of $L$ as a left ideal and also a \emph{suffix code}, i.e., no element of $G$ is proper suffix of another. 
\begin{proposition}\label{propositionLeftIdeal}
    Let $\mu$ be an invariant probability measure of $A^\Z$ and $L$ be a left ideal of $A^*$. Then
    $L$ has a density in the strong sense and $\delta_\mu(L)=\mu(G)$ where $G = L\setminus A^+L$.
\end{proposition}

\begin{proof}
    Let $X$ be the support of $\mu$. Since $G$ is a suffix code,
  \begin{equation*}
    \mu(A^*G\cap A^n)=\sum_{u\in G}\mu(A^*u\cap A^n)=
    \sum_{u\in G,|u|\le n}\mu(A^{n-|u|}u)=\sum_{u\in G,|u|\le n}\mu(u),
  \end{equation*}
  where the last equality uses the fact that since $\mu$ is invariant, $\mu(A^iu) = \mu(u)$ for all $i$. By invariance of $\mu$ again, this tends to $\mu(G)$ when $n\to\infty$.
\end{proof}

The density of a left ideal may not exist if the measure is not invariant, as shown next.

\begin{example}\label{exampleNotInvariant}
    Let $x\in A^\Z$ and let $\mu$ be the Dirac measure of $x$, that is, the probability measure $\mu$ on $A^\Z$ such that $\mu(U)=1$ if $x\in U$ and $0$ otherwise. Then, for $a\in A$, the density of $A^*a$ is the frequency of $a$ in the sequence $x^+=x_0x_1\cdots$. If the frequency of $a$ does not exist in $x^+$, the language $A^*a$ does not have a density.  For example, if $x^+=aba^2b^2\cdots a^{2^n}b^{2^n}\cdots$, then the frequency of $a$ does not exist. Indeed, the frequency of $a$ is $1/2$ in $aba^2b^2\cdots a^{2^n}b^{2^n}$, while it is close to $2/3$ in $aba^2b^2\cdots a^{2^n}b^{2^n}a^{2^{n+1}}$.
\end{example}

Next we consider the case of a \emph{quasi-ideal}, that is, the intersection of a left
ideal and a right ideal. In this case we need the stronger assumption that $\mu$ is ergodic.

\begin{proposition}\label{propositionQuasiIdeal}
    Let $\mu$ be an ergodic measure on $A^\Z$, $L$ be a right ideal of $A^*$, and $K$ be a left ideal of $A^*$. Then $L\cap K$ has a density and $\delta_\mu(L\cap K)=\mu(D)\mu(G)$ where $D = L\setminus LA^+$ and $G = K\setminus A^+K$. Moreover, if $\mu$ is mixing, then $L\cap K$ has a density in the strong sense.
\end{proposition}

\begin{proof}
    Let $X$ be the support of $\mu$. Fix $u\in D$ and $v\in G$ and let $m=\max(|u|,|v|)$. Then $\mu(uA^*\cap A^*v\cap A^n)=0$ if $n<m$ and otherwise
    \begin{displaymath}
        \mu(uA^*\cap A^*v\cap A^n)=\mu([u]_X\cap S^{|v|-n}[v]_X)
    \end{displaymath}
    and thus using Equation \eqref{eqErgodic},
    \begin{align*}
        \delta_\mu(uA^*\cap A^*v)=&\lim_{n\to\infty}\frac{1}{n}\sum_{i=0}^{n-1}\mu(uA^*\cap A^*v\cap A^i) 
        = \lim_{n\to\infty}\frac{1}{n}\sum_{i=m}^{n-1}\mu([u]_X\cap S^{|v|-i}[v]_X)\\
        &= \lim_{n\to\infty}\frac{1}{n}\sum_{i=0}^{n-1}\mu([u]_X\cap S^{-i}[v]_X)
        =\mu([u]_X)\mu([v]_X)=\mu(u)\mu(v).
    \end{align*}
    This shows that
    \[
        \delta_\mu(DA^*\cap A^*G)=\sum_{u,v}\delta_\mu(uA^*\cap A^*v)=\sum_{u,v}\mu(u)\mu(v)=\mu(D)\mu(G).
    \]
    If $\mu$ is mixing, then with similar arguments
    \[
        \lim_{n\to\infty} \mu(uA^*\cap A^*v\cap A^n)=\lim_{n\to\infty}\mu([u]_X\cap S^{-n}[v]_X)=\mu(u)\mu(v),
    \]
    which shows that the density exists in the strong sense. 
\end{proof}

\begin{example}\label{ex:not-mixing}
    If $X=\{(ab)^\infty, (ba)^\infty\}$, then the language $L=aA^*\cap A^*b$ has density $1/2$ with respect to the unique invariant measure $\mu$ on $X$. However $\mu$ is not mixing and the density does not exist in the strong sense as $\mu(aA^*\cap A^*b\cap A^n)$ alternates between 0 and 1.
\end{example}

Next we turn to \emph{two-sided ideals}, that is, languages $L$ such that $A^*LA^*=L$. 

\begin{proposition}\label{proposition2SidedIdeal}
    Let $\mu$ be an ergodic measure on $A^\Z$ with support $X$. Every two-sided ideal  $L$ that intersects $\cL(X)$ has density $1$ in the strong sense.
\end{proposition}

\begin{proof}
    Let $D = L\setminus LA^+$ and $G = L\setminus A^+L$. Then $L = DA^* = A^*G$ and by the formulas for right and left ideals, the density exists in the strong sense and $\delta_\mu(L)=\mu(D)=\mu(G)$. But, using the formula for quasi-ideals, we also have $\delta_\mu(L) = \mu(D)\mu(G) = \delta_\mu(L)^2$, so $\delta_\mu(L) = 0$ or 1. The fact that $L\cap \cL(X)\neq\varnothing$ implies $\delta_\mu(L)>0$. 
\end{proof}

Let $X$ be a shift space. We derive from \cref{proposition2SidedIdeal} the following property for $X$-thin codes, where  the languages $L$ such that there exists $w\in \cL(X)$  satisfying $ A^*wA^*\cap L=\varnothing$ are called \emph{$X$-thin} by the  terminology of \cite{BerstelDeFelicePerrinReutenauerRindone2012}.  A suffix code $C\subseteq \cL(X)$ is \emph{$X$-maximal} if it is not properly included in a suffix code $D\subseteq \cL(X)$. A dual notion holds for prefix codes. For example, every $\cL_n(X)$, for $n\ge 1$, is an $X$-maximal suffix code.
\begin{proposition}\label{propositionThinSuff}
    Let $\mu$ be an ergodic measure with support $X$ and let $C\subseteq\cL(X)$ be an $X$-maximal prefix or suffix code. If $C$ is $X$-thin, then $\mu(C)=1$.  
\end{proposition}

\begin{proof}
We assume  that $C$ is an $X$-maximal suffix code.
    Let $w\in\cL(X)$ be such that $A^*wA^*\cap C=\varnothing$. We claim that 
        $A^*wA^*\cap \cL(X)\subseteq A^*C$.
    Indeed, let $u,v$ be such that $uwv\in \cL(X)$. Since $C$ is $X$-maximal, either $uwv$ has a suffix in $C$ or is a suffix of some $c\in C$. The second case being impossible, we conclude that $uwv\in A^*C$, which proves the claim.  Therefore, we have $\delta_\mu(A^*wA^*)\le\delta_\mu(A^*C)$. Then by Propositions~\ref{proposition2SidedIdeal} and \ref{propositionLeftIdeal}, we have $\mu(C) =  \delta_\mu(A^*C)\geq \delta_\mu(A^*wA^*) = 1$.
    The proof for the dual statement for prefix codes works analogously.
\end{proof}
In particular, a finite prefix or prefix code is $X$-maximal if and only if $\mu(C)=1$ for an ergodic measure with support $X$. 

\section{J-class of a shift space}
\label{sec:Jclass}

This section is devoted to $\JJ$-classes of finite monoids associated with shift spaces. The reader who is not familiar with semigroup theory and Green's relations might want to consult textbooks such as~\cite{BerstelPerrinReutenauer2009,book/Lallement1979,book/Grillet1995,book/Howie1995,Eilenberg1974,Eilenberg1976}.
Alternatively, \cref{appendix-semigroups} provides a brief account of the necessary material.
The notion of a $\JJ$-class discussed here also appeared in \cite{Perrin2015,PerrinSchupp1986}, and is tangentially related with the work of Almeida on free profinite semigroups~\cite{Almeida2005}. See also the survey \cite{Colcombet2011} for more on the connections between Green's relations and automata theory.

\subsection{Definition and first properties}

Let $X$ be a shift space on $A$ and let $\varphi\colon A^*\to M$ be a morphism onto a finite monoid $M$. We first introduce two subsets of $M$ naturally associated with $X$. 

\begin{definition}
    Let $K_X(M)$ be the intersection of all two-sided ideals $I$ of $M$ such that $I\cap\varphi(\cL(X))\neq\varnothing$, called the \emph{$X$-minimal ideal of $M$}. Let $J_X(M)$ be the set of elements of $M$ which generate $K_X(M)$ as a two-sided ideal, i.e., 
    \begin{equation*}
        J_X(M) = \{ m\in M \mid MmM = K_X(M)\}.
    \end{equation*}
\end{definition}

It follows from the definition that $J_X(M)$ is a $\JJ$-class. Consider the quasi-order defined on $M$ by $m\leq_{\JJ} n\iff MmM\subseteq MnM$. Note that $m\JJ n$ if and only if $m\leq_{\JJ} n$ and $n\leq_{\JJ} m$. The next proposition establishes some basic properties of $J_X(M)$ when $X$ is irreducible.

\begin{proposition}\label{prop:Jclass}
    Let $X$ be an irreducible shift space on $A$ and let $\varphi\colon A^*\to M$ be a morphism onto a finite monoid $M$. Then 
    \begin{enumerate}
        \item $K_X(M)$ is an ideal  of $M$ which meets $\varphi(\cL(X))$;
            \label{Jclass-ideal}
        \item $J_X(M)=\{m\in K_X(M)\mid MmM\cap\varphi(\cL(X))\ne\varnothing\}$;
            \label{Jclass-intersect}
        \item $J_X(M)\cap\varphi(\cL(X))$ is the non-empty set of $\leq_{\JJ}$-minimal elements of $\varphi(\cL(X))$;
            \label{Jclass-J-minimal}
        \item Either $J_X(M)$ is the minimal ideal $K(M)$ of $M$, or $J_X(M)\cup\{0\}$ is the unique 0-minimal ideal in the quotient of $M$ by the largest ideal of $M$ which does not meet $\varphi(\cL(X))$. 
            \label{Jclass-min-ideal}
    \end{enumerate}
\end{proposition}

\begin{proof}
    \ref{Jclass-ideal}. Let $I,J$ be two ideals which meet $\varphi(\cL(X))$. Let $u,v\in\cL(X)$ be such that $\varphi(u)\in I$ and $\varphi(v)\in J$. Since $X$ is irreducible, there is a word $w$ such that $uwv\in\cL(X)$. Then $\varphi(uwv)\in I\cap J$. This proves that $K_X(M)$ is an ideal which meets $\varphi(\cL(X))$.
    
    \ref{Jclass-intersect}. Let $m\in J_X(M)$ and let $n\in K_X(M)\cap\varphi(\cL(X))$. Then $n$ is in $MmM$. This proves the inclusion from left to right. Conversely, let $m\in K_X(M)$ be such that $MmM\cap\varphi(\cL(X))\ne\varnothing$. Then, the ideal $MmM$ generated by $m$ is contained in $K_X(M)$ and it meets $\varphi(\cL(X))$. This implies $MmM=K_X(M)$ and therefore $m$ is in $J_X(M)$.

    \ref{Jclass-J-minimal}. Let $\varphi(u)$ be $\leq_{\JJ}$-minimal in $\varphi(\cL(X))$. Fix $v\in\cL(X)$. By irreducibility there is $w\in\cL(X)$ such that $uwv\in\cL(X)$. It follows that $\varphi(uwv)\leq_{\JJ}\varphi(u)$, which by minimality implies $\varphi(u)\leq_{\JJ}\varphi(uwv)\leq_{\JJ}\varphi(v)$. This shows that $\varphi(u)$ generates an ideal contained in every ideal which meets $\cL(X)$, thus $\varphi(u)\in J_X(M)$.
   
    \ref{Jclass-min-ideal}. First assume that $K(M)\cap\varphi(\cL(X))\neq\varnothing$. Since the elements of $K(M)$ are $\JJ$-minimal in all of $M$, it follows from part \ref{Jclass-J-minimal} that $J_X(M) = K(M)$.

    Assume next that $K(M)\cap\varphi(\cL(X))=\varnothing$. Let $I=\{s\in M \mid MsM\cap\varphi(\cL(X))=\varnothing\}$. It is clear that $I$ is the largest ideal of $M$ which does not meet $\varphi(\cL(X))$. It is non-empty since it contains the minimal ideal $K(M)$ by assumption. Let us show that $M/I$ is a prime monoid (see \cite[Section~1.12]{BerstelPerrinReutenauer2009}). Take $s,t\in M/I$ which are $\neq 0$. Then there exist $x_1,x_2,y_1,y_2$ such that $x_1sx_2 = \varphi(u)$ and $y_1ty_2 = \varphi(v)$, where $u,v\in\cL(X)$. By irreducibility there is $w$ such that $uwv\in\cL(X)$, hence $x_1sx_2\varphi(w)y_1ty_2 = \varphi(uvw)$, and in particular $sx_2\varphi(w)y_1t\neq 0$. This shows that $M/I$ is prime, and thus it admits a unique 0-minimal ideal, by \cite[Proposition 1.12.9]{BerstelPerrinReutenauer2009}. Finally notice that any non-zero element $s$ of the quotient $M/I$ must satisfy $xsy = \varphi(u)$ for some $x,y\in M$, thus it is $\JJ$-above some element of $J_X(M)$. This implies that $J_X(M)\cup\{0\}$ is the 0-minimal ideal of $M/I$.
\end{proof}

\begin{example}\label{ex:Jclass}
    Let $\A$ be the automaton depicted, alongside its transition monoid, in \cref{figureMonoidGolden} and $X$ be the Fibonacci shift (\cref{exampleFibonacci}). Let $\varphi\from A^*\to M$ be the transition morphism of $\A$. Then $\varphi^{-1}(0) = A^*bbA^*$, and since $\cL(X)$ does not contain $bb$, it follows that $J_X(M)$ is the $\JJ$-class of $\alpha=\varphi(a)$. In this example, the automaton $\A$ is the minimal automaton of the $X$-maximal prefix code $C=\{a,ba\}$.
    \begin{figure}
        \centering
        \tikzstyle{loop above}=[in=50,out=130,loop]
            \begin{tikzpicture}
                \node[state](1)at(0,0){$1$};
                \node[state](2)at(2,0){$2$};

                \draw[->,label,loop left,swap](1)edge node{$a$}(1);
                \draw[->,label,bend left](1)edge node{$b$}(2);
                \draw[->,label,bend left](2)edge node{$a$}(1);
            \end{tikzpicture}
            \qquad
            \begin{tikzpicture}
                
                \matrix[] (j1) [dclass]{
                    1\\
                };
                \matrix[right=20pt of j1] (j) [dclass]{
                    $\alpha$ & $\alpha\beta$ \\
                    $\beta\alpha$ & $\beta$ \\
                };
                \drawgrid{2}{2}{j}
                \matrix[right=20pt of j] (j0) [dclass]{
                    $0$ \\
                };
                \node[idempotent] at (j0-1-1.north east) {*};
                \node[idempotent] at (j1-1-1.north east) {*};
                \node[idempotent] at (j-1-1.north east) {*};
                \node[idempotent] at (j-1-2.north east) {*};
                \node[idempotent] at (j-2-1.north east) {*};
            \end{tikzpicture}
            \caption{An automaton $\A$ and the eggbox picture of its transition monoid $M$, where $\alpha$ represents the transformation of the states induced by $a$, and $\beta$ the one induced by $b$.}
            \label{figureMonoidGolden}
        \end{figure}
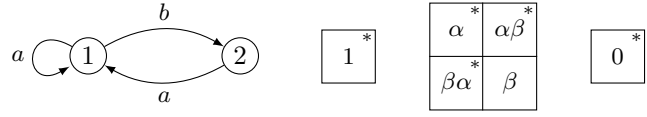
\end{example}

The following result links the $\JJ$-class $J_X(M)$ with the densities of the languages recognized by $M$. Roughly speaking, it shows that the $\JJ$-class $J_X(M)$ concentrates the density. 

\begin{proposition}\label{p:density-concentration}
    Let $\mu$ be an ergodic measure on $A^\Z$ with support $X$. Let $\varphi\from A^*\to M$ be a morphism onto a finite monoid $M$. Then the density of $L = \varphi^{-1}(J_X(M))$ exists in the strong sense and $\delta_\mu(L) = 1$.
\end{proposition}

\begin{proof}
    It follows from \cref{prop:Jclass} that $K_X(M)\cap\varphi(\cL(X)) =  J_X(M)\cap\varphi(\cL(X))$, hence $\delta_\mu(\varphi^{-1}(J_X(M)))=\delta_\mu(\varphi^{-1}(K_X(M)))$. But the density of $K_X(M)$ exists in the strong sense and is $1$ by \cref{proposition2SidedIdeal}.
\end{proof}

Let us highlight a simple consequence of this.
\begin{corollary}\label{c:density-0}
    Let $\mu$ be an ergodic measure with support a shift space $X$. Let $\varphi\from A^*\to M$ be a morphism onto a finite monoid $M$. For every $m\notin J_X(M)$, the density of $\varphi^{-1}(m)$ exists in the strong sense and is 0. 
\end{corollary}

\begin{proof}
    Indeed, let $L = \varphi^{-1}(m)$, $L'=\varphi^{-1}(J_X(M))$. Then as $L\subseteq A^*\setminus L'$ and the density of $L'$ has density 1 in the strong sense,
    \begin{equation*}
        0= 1-\lim_{n\to\infty}\mu(L'\cap A^n) =\lim_{n\to\infty}\mu(A^n\setminus L') \geq  \lim_{n\to\infty}\mu(L\cap A^n) = \delta_\mu(L).\qedhere
    \end{equation*}
\end{proof}

\subsection{On aperiodic languages} \label{subsec:aperiod}
The next theorem concerns the density of aperiodic languages. Recall that a language is \emph{aperiodic} if it can be recognized by an aperiodic monoid. By Schützenberger's theorem~\cite{Schutzenberger1965a}, aperiodic languages are precisely the \emph{star-free} languages, as well as the languages defined by first-order formulas (see \cite[Theorem~4.1]{book/Perrin2004}). The density of aperiodic languages was also studied, from a logic perspective, by Lynch~\cite{Lynch1993}. Lynch's definition of density involves sequences of Bernoulli measures, thus it is both more general than our definition (since we use only one measure) and much more particular, since we use more general measures than Bernoulli measures.

\begin{theorem}\label{theoremDensityAperiodic}
    An aperiodic language $L$ has a density with respect to an ergodic measure $\mu$.  More precisely, there exist a finite number of pairs $(D_i,G_i)_{1\le i\le k}$ of a prefix code and a suffix code such that $L\cap\cL(X)=\bigcup_{i=1}^k(D_iA^*\cap A^*G_i)\cap \cL(X)$ and
    \begin{equation}
        \delta_\mu(L)=\sum_{i=1}^k\mu(D_i)\mu(G_i).\label{eqDensityAperiodic}
    \end{equation}
    Moreover, if the measure is mixing, then the density exists in the strong sense.
\end{theorem}

\begin{proof}
    Let $\mu$ be an ergodic measure and $\varphi\colon A^*\to M$ be a morphism onto an aperiodic monoid $M$. Let $J=J_X(M)$, let $m\in M$ and let $L=\varphi^{-1}(m)$. It is enough to prove that Equation \eqref{eqDensityAperiodic}
    holds for such $L$. If $m\notin J$, then $\delta_\mu(L)=0$ in the strong sense by \cref{c:density-0}, so we may assume that $m\in J$.  We claim that 
    \[
        L\cap \cL(X)=LA^*\cap A^*L\cap \cL(X).
    \]
    Indeed, assume that $u\in(LA^*\cap A^*L)\cap\cL(X)$. Then $\varphi(u) \in mM \subseteq K_X(M)M = K_X(M)$, and thus $\varphi(u)\in J$ by the second part of \cref{prop:Jclass}. But then $\varphi(u)\in mM\cap Mm\cap J$ and since $M$ is aperiodic, we have $mM\cap Mm\cap J = \{m\}$. It follows that $\varphi(u)=m$ and $u\in L$. This concludes the proof of the claim, as the other inclusion is trivial. Finally, we can use \cref{propositionQuasiIdeal} to deduce that $\delta_\mu(L)=\mu(D)\mu(G)$, where $D=L\setminus LA^+$ and $G=L\setminus A^+L$, and that the density exists in the strong sense if $\mu$ is mixing. 
\end{proof}

Let us give two examples to illustrate   \cref{theoremDensityAperiodic}, starting with an aperiodic language.
\begin{example}\label{example(ab+ba)*}
    Let $L=\{ab,ba\}^*$, whose minimal automaton $\A$ may be found in \cite[Chapter~4, Example~2.1]{book/Pin1986}.  Let $M$ stand for its transition monoid. Note that $L$ is aperiodic and thus star-free, which is not obvious. A star-free expression of $L$ may be found in \cite[Chapter~4, Example~2.1]{book/Pin1986}. Let $X$ be Thue--Morse shift and let $\mu$ be the unique invariant measure on $X$ (\cref{exampleMorse}). The $\JJ$-class $J_X(M)$ is the $\JJ$-class of $\alpha^2$ in \cite[Chapter~4, Example~2.1]{book/Pin1986}.
    Combining Equation \eqref{eqDensityAperiodic} with the expression of $L$ in \cite[Chapter~4, Example~2.1]{book/Pin1986}, we get 
    \begin{equation*}
        \delta_\mu(L)=\mu((ab)^+b\cup (ba)^+a)\mu(a(ab)^+\cup b(ba)^+).
    \end{equation*}
    With the values of $\mu$ from \cite[Example 3.8.20]{DurandPerrin2021}, we find 
    \begin{equation*}
        \mu((ab)^+b\cup (ba)^+a)=\mu(\{abb,ababb,baa,babaa\})=1/2
    \end{equation*}
    and similarly $\mu(a(ab)^+\cup b(ab)^+)=1/2$. Thus $\delta_\mu(L)=1/4$.
\end{example}

The next example is a group language whose intersection with the shift space behaves like an aperiodic language. 

\begin{example}\label{exampleKarl}
    Let $X$ be the orbit of the periodic sequence $x=(abc)^\infty$. Thus $X$ consists of three elements and has for unique ergodic measure the uniform probability measure $\mu$. Let $\varphi\colon A^*\to\Z/2\Z$ be defined by $\varphi(a)=0$ and $\varphi(b)=\varphi(c)=1$. Consider the rational language $L=\varphi^{-1}(0)$. Then one sees
    \begin{displaymath}
        L\cap\cL(X)=(abc)^*\{\varepsilon,a\}\cup(bca)^*\{\varepsilon,bc\}\cup (cab)^*\{\varepsilon\},
    \end{displaymath}
    and thus $\delta_\mu(L)=\frac{1}{3}\left(1+\frac{1}{3}+\frac{1}{3}\right)=\frac{5}{9}$.

    The same result can be obtained using \cref{theoremDensityAperiodic}. First, observe that $L$ has the same intersection with $\cL(X)$ as the language $B^*$ where $B=\{a,bc,cab\}$. The language $B^*$ is recognized by the automaton $\A$ depicted in \cref{figureExampleKarl}. The transition monoid of $\A$ is an aperiodic monoid $M$ with 27 elements, hence  $B^* $ turns out to be  star-free. The $\JJ$-class $J_X(M)$, depicted in \cref{figureExampleKarl}, consists of 9 elements, 5 of which belong to the image of $B^*$ (indicated in yellow). Taking for instance $m = \beta\gamma=\varphi(bc)$, we find that $\varphi^{-1}(m)A^* = bcA^*$ and $A^*\varphi^{-1}(m) = A^*bc$, and thus $\delta_\mu(\varphi^{-1}(m))=\mu(bc)^2=1/9$ by Equation \eqref{eqDensityAperiodic}; likewise $\delta_\mu(\varphi^{-1}(m))=1/9$ for all other elements $m\in J_X(M)$ using Equation \eqref{eqDensityAperiodic}. Thus we recover the fact that $\delta_\mu(B^*)=\delta_{\mu}(B^* \cap \cL(X))=\delta_{\mu}(L \cap \cL(X))=5/9$.

    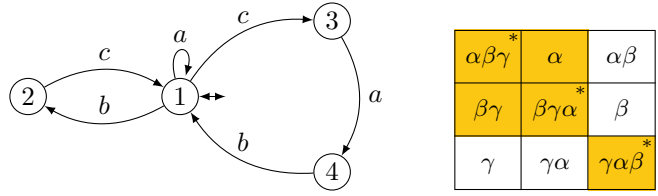
\begin{figure}
      \centering
        \begin{tikzpicture}
          \node[state](1)at(0,0){$1$};\draw[<->](1)edge node{}(0.6,0);
          \node[state](2)at(-2,0){$2$};
          \node[state](3)at(2,1){$3$};
          \node[state](4)at(2,-1){$4$};

          \draw[above,->](1)edge [loop above]node{$a$}(1);
          \draw[above,->,bend left](1)edge node{$b$}(2);
          \draw[above,->,bend left](2)edge node{$c$}(1);
          \draw[above,->,bend left](1)edge node{$c$}(3);
          \draw[right,->,bend left](3)edge node{$a$}(4);
          \draw[above,->,bend left](4)edge node{$b$}(1);
        \end{tikzpicture}
        \qquad
        \begin{tikzpicture}
            \matrix[dclass,minimum height=20pt,minimum width=25pt] (j){
                \strut$\alpha\beta\gamma$ & \strut$\alpha$ & \strut$\alpha\beta$ \\
                $\beta\gamma$  & $\beta\gamma\alpha$ & $\beta$ \\ 
                \strut$\gamma$  & \strut$\gamma\alpha$ & \strut$\gamma\alpha\beta$ \\ 
            };
            \drawgrid{3}{3}{j}
            \node[idempotent] at (j-1-1.north east) {*};
            \node[idempotent] at (j-2-2.north east) {*};
            \node[idempotent] at (j-3-3.north east) {*};
            \scoped[on background layer]
            {
                \node[fill=lipicsYellow,fit=(j-1-1),inner sep=0pt]   {};
                \node[fill=lipicsYellow,fit=(j-1-2),inner sep=0pt]   {};
                \node[fill=lipicsYellow,fit=(j-2-1),inner sep=0pt]   {};
                \node[fill=lipicsYellow,fit=(j-2-2),inner sep=0pt]   {};
                \node[fill=lipicsYellow,fit=(j-3-3),inner sep=0pt]   {};
            }
        \end{tikzpicture}
        \caption{The minimal automaton of $B^*$ and the $\JJ$-class $J_X(M)$ from \cref{exampleKarl}, where $\alpha$, $\beta$ and $\gamma$ are the transformations induced respectively by $a$, $b$ and $c$. The $\HH$-classes in yellow represent elements in the image of $B^*$.}\label{figureExampleKarl}
    \end{figure}
 \end{example}
 
\subsection{Decidability and degree}

Next we consider the question of whether membership in $J_X(M)$ is decidable. We say that a shift space $X$ is \emph{rationally decidable} if the emptiness of $\cL(X)\cap L$ is decidable for every rational language $L$. Sofic shifts are obviously rationally decidable. The class of rationally decidable shift spaces also includes all substitution shifts by a result of~\cite[Lemma 3]{Salo2017} (see also~\cite[Lemma 3.15]{BealPerrinRestivo2024a} and \cite{carton02}).
\begin{proposition} \label{prop:decidable}
    Let $X$ be a rationally decidable shift space. Then for every morphism $\varphi\colon A^*\to M$, there is an algorithm which decides which elements $m\in M$ are in $J_X(M)$.
\end{proposition}

\begin{proof}
    We have $m\in K_X(M)$ if and only if $m\in MnM$ for every $n\in M$ such that $\varphi^{-1}(n)\cap \cL(X)\ne\varnothing$. Since $X$ is rationally decidable, this is decidable for every $n\in M$. Therefore, we can decide whether $m\in K_X(M)$, and then whether $m\in J_X(M)$.
\end{proof}

The $\JJ$-class is also linked with the following notion of degree. Let $\A$ be an automaton on $A$ and $\varphi\from A^*\to M$ be its transition morphism. The minimal rank of the elements of $\varphi(\cL(X))$ viewed as partial mappings is called the \emph{$X$-degree} of the automaton, denoted $d_X(\A)$. Note that \cite[Proposition~3.2]{Perrin2015} states that $J_X(M)$ contains all elements of rank $d_X(\A)$. For instance the automaton $\A$ from \cref{exampleKarl} has $X$-degree 2, while the one in \cref{ex:Jclass} has $X$-degree 1. This result also implies that $d_X(\A)$ is computable if and only if membership in $J_X(M)$ is decidable (cf.  \cref{prop:decidable} for cases whether the latter is decidable).

The notion of $X$-degree has also been studied in relation with codes~\cite{BerstelDeFelicePerrinReutenauerRindone2012,Almeida2020}. Let $X$ be an irreducible shift, let $C\subseteq\cL(X)$ be a rational prefix code and let $\A$ be the minimal automaton of $C^*$.  If $C$ is $X$-maximal, then its $X$-degree (that is, $d_X(\A)$) is larger than or equal to $1$. When $C$ is a finite $X$-maximal bifix code (i.e. both prefix and suffix),  then the $X$-degree of $C$ can also be defined in terms of the number of \emph{parses} of elements of  $\cL(X)$~\cite{BerstelDeFelicePerrinReutenauerRindone2012}. 

\section{Existence of  the density}
\label{sec:existence}
    
The goal of this section is to establish the existence of density of regular languages under all invariant measures. More precisely, we will prove the following, which is our main theorem. 
\begin{theorem}\label{theoremMain}
    Let $\mu$ be an invariant measure on $A^\Z$. Then every rational language $L$ has a density with respect to $\mu$. 
\end{theorem}

The main ingredient for the proof is the dynamical system defined as follows, called a skew product. Let $X$ be an irreducible shift space and $\varphi\colon A^*\to M$ a morphism onto a finite monoid. Fix $R$ an $\RR$-class of $J_X(M)$ such that $R\cap\varphi(\cL(X))\neq\varnothing$. Let $M$ act on the right of $R\cup\{0\}$ by $r\cdot m = rm$ if $rm\in R$ and $r\cdot m=0$ otherwise. 
The  specific choice  of  an $\RR$-class class $R$ plays no role here, as long as $R$ satisfies $R\cap\varphi(\cL(X))\neq\varnothing$.

\begin{definition}
    The skew product $R\cup\{0\}\rtimes X$ is the system $((R\cup\{0\})\times X,T)$ where  
    \begin{equation*}
        T(r,x) = (r\cdot \varphi(x_0), Sx).
    \end{equation*}
\end{definition}

Note that the projection $\pi_X\from(R\cup\{0\})\times X\to X$ satisfies $\pi_X\circ T = S\circ\pi_X$. It follows that if $\nu$ is an invariant (respectively ergodic) measure on $(R\cup\{0\})\rtimes X$, then $\nu\circ\pi_X^{-1}$ is an invariant (respectively ergodic) measure on $X$.  When $M=G$ is a group, then $R=G$ and $G\rtimes X$ forms a subsystem of $(G\cup\{0\})\rtimes X$, which in effect means we can get rid of 0. In particular we recover the type of skew products studied in the preprint~\cite{BertheGouletOuelletNybergBroddaPerrinPetersen2024}. Next, we introduce a natural probability measure on $(R\cup\{0\})\rtimes X$ induced by a measure $\mu$ on $X$. In the group case, we recover the product of $\mu$ with the normalized counting measure on the group considered in~\cite{BertheGouletOuelletNybergBroddaPerrinPetersen2024}.

\begin{definition}
    Fix an invariant probability measure $\mu$ on $X$. The \emph{weighted counting measure} on the skew product $(R\cup\{0\})\rtimes X$ is the Borel measure $\nu$ defined by $\nu(\{0\}\times X)=0$, and for every $u,v\in\cL(X)$ and $r\in R$,
    \[
        \nu(\{r\}\times [u\cdot v]_X)=\frac1d\sum_{s,s\varphi(u)=r}\mu(G_suv),
    \]
    where $d$ is the cardinality of the $\HH$-classes of $J_X(M)$ and $G_s$, for $s\in R$, is the suffix code such that $\varphi^{-1}(Ms)=A^*G_s$.
\end{definition}

That $\nu$ is a well-defined Borel measure on $(R\cup\{0\})\rtimes X$ follows from Carathéodory's extension theorem, since the above formula defines a pre-measure on the Boolean algebra of clopen sets of $(R\cup\{0\})\times X$. Note that if $r\HH s$, then $G_r = G_s$. For an $\HH$-class $H\subseteq R$, let $G_H$ be the common value of $G_r$ for $r\in H$. The following technical lemma will be useful. In its proof, we use the following notion: we call a monoid $M$ \emph{stable} if for every $s,t\in M$, the implications 
    $s\JJ st \implies s\RR st$ and 
    $s\JJ ts \implies s\LL ts$
hold. Every finite monoid is stable~\cite[Lemma 1.1, Chapter V]{book/Grillet1995}.

\begin{lemma}\label{lemmaGr}
    The union $G = \bigcup_{H\subseteq R}G_H$ over all $\HH$-classes of $R$ is a suffix code such that $\mu(G)=1$.
\end{lemma}

\begin{proof}
    The fact that $G$ is a suffix code is a consequence of the fact that the monoid $M$ is stable. Moreover, $G$ is $X$-maximal since, for all $w\in\cL(X)$, there exists $u$ such that $\varphi(uw)\in J_X(M)$, by irreducibility and the fact that $R\cap\varphi(\cL(X))\neq\varepsilon$. Note that a word $w\in \cL(X)$ such that $\varphi(w)\in J_X(M)$ cannot be a  proper factor of any element of $G$, since $G$ is a suffix code. Thus there must exist $w\in\cL(X)$ such that $A^*wA^*\cap G = \varnothing$ (simply fix $u\in G$ and then take $w\in \cL(X)\cap A^+u$). It follows that $\mu(G)=1$ by \cref{propositionThinSuff}.
\end{proof}

Next we establish some properties of the weighted counting measure.

\begin{proposition}\label{p:weighted-counting}
    The weighted counting measure is an invariant probability measure on the skew product $(R\cup\{0\})\rtimes X$ which satisfies $\nu\circ\pi_X^{-1} = \mu$.
\end{proposition}

\begin{proof}
    Next we show that $\nu$ is invariant under the map $T$. Note that it suffices to check invariance for sets of the form $B = \{r\}\times [u\cdot v]_X$. First, assume that $u\neq\varepsilon$. If we let $u = u'a$, $a\in A$, then $T^{-1}(\{r\}\times[u\cdot v]_X)$ can be expressed as the disjoint union
    \[
        T^{-1}(\{r\}\times[u\cdot v]_X) = \bigcup_{s,s\varphi(a)=r} \{s\}\times[u'\cdot av]_X,
    \]
    and from the definition of $\nu$ we find
    \begin{align*}
        \nu(T^{-1}(\{r\}\times[u\cdot v]_X)) 
        &= \sum_{s,s\varphi(a)=r}\nu(\{s\}\times[u'\cdot av]_X) 
        = \sum_{s,s\varphi(a)=r}\frac1d\sum_{s',s'\varphi(u')=s}\mu(G_{s'}u'av) \\
        &= \frac1d\sum_{s,s\varphi(u)=r}\mu(G_{s}uv) = \nu(\{r\}\times[u\cdot v]_X).
    \end{align*}
    Next let us treat the case where $u=\varepsilon$, or in other words where $B = \{r\}\times[v]_X$. Observe that 
    \[
        T^{-1}(\{r\}\times[v]_X)=\bigcup_{s\varphi(a)=r}\{s\}\times [av]_X, \qquad [G_rv\cdot \varepsilon]_X = \bigcup_{s\varphi(a)=r}[G_sav\cdot \varepsilon]_X,
    \]
    where the unions are taken over pairs $(s,a)\in R\times A$ such that $s\varphi(a)=r$. Since $G=\bigcup_{s\in R}G_s$ is a suffix code by \cref{lemmaGr}, the second union in the above equation is disjoint (as is the first, for obvious reasons). Thus we have
    \begin{align*}
        \nu(T^{-1}(\{r\}\times [v]_X))&=\sum_{s\varphi(a)=r}\nu(\{s\}\times[av]_X)=\sum_{s\varphi(a)=r}\frac1d\mu(G_sav)\\
                                    &=\frac1d\mu(G_rv)=\nu(\{r\}\times[v]_X). 
    \end{align*}

    Finally, let us show that $\nu(R\times U) = \mu(U)$ for every Borel set $U$. Using \cref{lemmaGr} together with the invariance of $\mu$, we have
    \begin{equation*}
        \nu(R\times U) = \sum_{r\in R}\mu([G_r\cdot \varepsilon]_X\cap U)/d = \sum_{H\subseteq R}\mu([G_H\cdot \varepsilon]_X\cap U) = \mu([G\cdot\varepsilon]_X\cap U) = \mu(U),
    \end{equation*}
    where the last equality follows since $\mu([G\cdot\varepsilon]_X) = \mu(G) = 1$. Taking $U=X$, we find that $\nu(R\times X)=\mu(X)=1$, which shows that $\nu$ is indeed a probability measure.
\end{proof}

The following lemma will also be needed for the proof of \cref{theoremMain}. Its proof uses standard arguments from ergodic theory. We include it for the sake of completeness.

\begin{lemma}\label{l:existence-project}
    For every ergodic measure $\mu$ on $X$, there exists an ergodic measure $\bar\mu$ on $(R\cup\{0\})\rtimes X$ such that $\bar\mu(\{0\}\times X) = 0$ and $\bar\mu\circ\pi_X^{-1}=\mu$.
\end{lemma}

\begin{proof}
    Let $\mathcal{M}$ be the set of all invariant probability measures on $(R\cup\{0\})\rtimes X$ and $\mathcal{M}_\mu$ the subset of those measures $\zeta$ such that $\zeta(\{0\}\times X)=0 $ and $\zeta\circ\pi_X^{-1} = \mu$. Note that $\mathcal{M}_\mu$ contains the weighted counting measure by \cref{p:weighted-counting} and thus $\mathcal{M}_\mu \neq\varnothing$. Moreover, $\mathcal{M}_\mu$ is a closed convex subspace of $\mathcal{M}$, so it must contain an extreme point $\bar\mu$ by the Krein--Milman theorem. Let us show that $\bar\mu$ is also an extreme point of $\mathcal{M}$. Suppose that $\bar\mu = s\zeta'+(1-s)\zeta''$ for some $\zeta',\zeta''\in\mathcal{M}$ and $0<s<1$. It is clear that both $\zeta'$ and $\zeta''$ must give  to $\{0\}\times X$ zero measure, and the fact that $\mu$ is an extreme point in the convex set of invariant measures on $X$ implies that $\zeta'\circ\pi_X^{-1}=\zeta''\circ\pi_X^{-1}=\mu$. Thus $\zeta',\zeta''\in\mathcal{M}_\mu$, contradicting the fact that $\bar\mu$ is an extreme point of $\mathcal{M}_\mu$. Since the extreme points of $\mathcal{M}$ are precisely the ergodic measure on $(R\cup\{0\})\rtimes X$, we are done. 
\end{proof}

The next proposition is a more precise form of \cref{theoremMain} for the case where $\mu$ is ergodic. Note that the statement uses \cref{l:existence-project} implicitly
for the existence of $\bar\mu$.
\begin{proposition}\label{propositionDensityMonoid} 
    Let $\mu$ be an ergodic measure on $A^\Z$ with support a shift space $X$. Let $\varphi\colon A^*\to M$ be a morphism onto a finite monoid and let $L=\varphi^{-1}(m)$ for some $m\in M$.  Fix an $\RR$-class $R$ of the $\JJ$-class $J_X(M)$. We set  the notation $U_{r,V}=\{r\}\times V$. Then $\delta_\mu(L)=0$ if $m\notin J_X(M)$ and otherwise
    \begin{equation}
        \delta_\mu(L)= \sum_{r,rm\in R}\bar\mu(U_{r,[L]_X})\bar\mu(U_{rm,X})
        \label{eqDensityRationalLanguage}
    \end{equation}
    where $\bar\mu$ is any ergodic measure on $(R\cup\{0\})\rtimes X$ such that $\bar\mu(\{0\}\times X)=0$ and $\bar\mu\circ\pi_X^{-1} = \mu$.
\end{proposition}

\begin{proof}
    If $m\notin J_X(M)$, then $\delta_\mu(L)=0$ by \cref{c:density-0}. We may now assume that $m\in J_X(M)$. Let $C$ be the prefix code such that $LA^*=CA^*$. For $i\ge 0$, let 
    \begin{equation*}
        C_{\le i}=\{u\in C\mid |u|\le i\},\quad C_{>i}=\{u\in C\mid |u|>i\}.
    \end{equation*}
    
    We claim that for every $r\in R$ such that $rm\in R$, one has
    \[
        U_{r,[L\cap A^i]_X}=U_{r,[C_{\le i}]_X}\cap T^{-i}(U_{rm,X}).
    \]
    The left-to-right inclusion is obvious. For the converse, we take $x$ such that
    $(r,x)\in U_{r,[C_{\leq i}]_X}\cap T^{-i}(U_{rm,X})$. This means that there exists
    some  $j$, with $0\leq j\leq i$, such that $\varphi(x_0\cdots x_{j-1}) = m$ and furthermore $rm'
    = rm$ where $m' = \varphi(x_0\cdots x_{i-1})$. Clearly $m \geq_{\RR} m'$, hence $m\RR
    m'$ by \cref{prop:Jclass} and stability of $M$. Moreover our choice of $r$
    guarantees that $rm\JJ m$, thus $rm\LL m$ by stability. By Green's lemma, $x\mapsto
    rx$ is a bijection between the $\RR$-classes of $m$ and $rm$. Since $rm=rm'$ it
    follows that $m=m'$, i.e. $\varphi(x_0\cdots x_{i-1})=m$. Thus $x\in[L\cap A^i]_X$ which proves the claim.

    As a result, we have
    \[
        \mu(L\cap A^i)=\bar\mu(R\times[L\cap A^i]_X)=\sum_{r,rm\in R}\bar\mu\Bigl(U_{r,[C_{\le i}]_X}\cap T^{-i}(U_{rm,X})\Bigr).
    \]
    Next we claim that for $\epsilon >0$, there is $i_0\ge 0$ such that $\bar\mu(U_{r,[C_{>i_0}]_X})<\epsilon$ for every $r\in R$. First observe that
    \begin{equation*}
        \sum_{r\in R}\bar\mu(U_{r,[C_{>i_0}]_X}) = \bar\mu(R\times [C_{>i_0}]_X) = \mu([C_{>i_0}]_X),
    \end{equation*}
    since $\bar\mu$ projects to $\mu$. Moreover since $C$ is a prefix code, the cylinders $[c]_X$ for $c\in C$ are all disjoint, and thus we may write
    \begin{equation*}
        \mu([C_{>i_0}]_X) = \mu(C_{>i_0}) = \sum_{i>i_0}\sum_{c\in C\cap A^{i}}\mu(c).
    \end{equation*}
    But note that this is a tail of the series
    \begin{equation*}
        \sum_{i\geq 0}\sum_{c\in C\cap A^{i}}\mu(c) = \mu(C) \leq 1,
    \end{equation*}
    and thus the claim immediately follows.

    Using the above claim together with the ergodicity of $\bar\mu$, this gives
    \begin{align*} 
        \delta_\mu(L)&=\lim_{n\to\infty}\frac{1}{n}\sum_{i=i_0}^{n-1}\mu(L\cap A^i)
        =\sum_{r,rm\in R}\lim_{n\to\infty}\frac{1}{n}\sum_{i=i_0}^{n-1}\bar\mu\bigl(U_{r,[C_{\le i}]_X}\cap T^{-i}(U_{rm,X})\bigr)\\
        &\ge\sum_{r, rm\in R}\lim_{n\to\infty}\frac{1}{n}\sum_{i=i_0}^{n-1}\bar\mu\bigl(U_{r,[C]_X}\cap T^{-i}(U_{rm,X})\bigr)-\epsilon\\
        &\ge\sum_{r, rm\in R}\bar\mu(U_{r,[L]_X})\bar\mu(U_{rm,X})-\epsilon.
    \end{align*}
    On the other hand, we have
    \begin{align*}
        \delta_\mu(L)&=\sum_{r, rm\in R}\lim_{n\to\infty}\frac{1}{n}\sum_{i=i_0}^{n-1}\bar\mu\bigl(U_{r,[C_{\le i}]_X}\cap T^{-i}(U_{rm,X})\bigr)\\
        &\le\sum_{r, rm\in R}\lim_{n\to\infty}\frac{1}{n}\sum_{i=i_0}^{n-1}\bar\mu\bigl(U_{r,[C]_X}\cap T^{-i}(U_{rm,X})\bigr)
        \le\sum_{r, rm\in R}\bar\mu(U_{r,[L]_X})\bar\mu(U_{rm,X}),
     \end{align*}
    concluding the proof.
\end{proof}

In order to conclude the proof of \cref{theoremMain}, we need to deduce the existence
of densities for a general invariant measure. This is done using the ergodic decomposition
discussed at the end of \cref{subsec:top-dynamics}. 
\begin{proof}[Proof of \cref{theoremMain}]
    Let $L$ be a rational language and $\mu$ be an invariant measure. Let $\mathcal E$ be the set
    of ergodic measure on the support of $\mu$ and consider the ergodic
    decomposition of $\mu$ as in Equation \eqref{eqErgodicDecomposition}. By
    \cref{propositionDensityMonoid}, $\delta_\nu(L)$ exists for every 
    $\nu\in\mathcal{E}$. Finally by the dominated convergence theorem 
    \begin{align*}
        \delta_\mu(L)&=\lim_{n\to\infty}\frac{1}{n}\sum_{i=0}^{n-1}\mu(L\cap A^i) 
        =\lim_{n\to\infty}\frac{1}{n}\sum_{i=0}^{n-1}\int_{\mathcal E}\nu(L\cap A^i) \, d\tau(\nu)\\
        &=\int_{\mathcal E}\lim_{n\to\infty}\frac{1}{n}\sum_{i=0}^{n-1}\nu(L\cap A^i) \, d\tau(\nu)
        =\int_{\mathcal E}\delta_\nu(L) \, d\tau(\nu).\qedhere
    \end{align*}
\end{proof}

When the weighted counting measure is ergodic, we can apply
\cref{propositionDensityMonoid} to obtain formulas based on the density of ideals, which are easily computable (\cref{sec:idealsfullmonoid}). This generalizes a result of \cite{BerstelPerrinReutenauer2009} where the same formula is proved for a Bernoulli measure. 
\begin{theorem}\label{theoremWeightedCountingMeasure}
 Let $\mu$ be an ergodic measure on $A^\Z$ with support a shift space $X$. Let $\varphi\colon A^*\to M$ be a morphism onto a finite monoid.  Fix $R$ an $\RR$-class of $J_X(M)$.  If the weighted counting measure is an ergodic measure on the skew product $(R\cup\{0\})\rtimes X$, then, for every $m\in J$, the density of $L=\varphi^{-1}(m)$ is
  \begin{equation}
   \delta_\mu(L)=\delta_\mu(A^*L)\delta_\mu(LA^*)/d,\label{eqDensity}
  \end{equation}
    where $d$ is the cardinality of the $\HH$-classes of $J_X(M)$.
\end{theorem}

\begin{proof}
    Let $\nu$ be the weighted counting measure on $(R\cup\{0\})\rtimes X$. Let $C$ be the prefix code such that $LA^*=CA^*$. Then Equation \eqref{eqDensityRationalLanguage} reduces to
    \begin{align*}
      \delta_\mu(L)&=\sum_{r, rm\in R}\nu(U_{r,[L]_X})\nu(U_{rm,X})
            =\frac{1}{d^2}\sum_{r, rm\in R}\mu([G_r \cdot C]_X)\mu(G_{rm})\\
            &=\frac{1}{d^2}\mu(G_m)\sum_{r, rm\in R}\mu([G_r \cdot C]_X)
            =\frac{1}{d}\mu(G_m)\sum_{H\subseteq R}\mu([G_H  \cdot C]_X)
    \end{align*}
    where $H$ runs over the $\HH$-classes of $R$ and $G_H$ is the common value of $G_r$ for the $d$ elements $r\in H$. By \cref{lemmaGr} and since $C$ is a prefix code,
    \[
        \sum_{H\subseteq R}\mu([G_H \cdot C]_X)=\mu([C]_X) = \mu(C)
    \]
    and therefore by Propositions~\ref{propositionRightIdeal} and \ref{propositionLeftIdeal}
    \[
        \delta_\mu(L)=\frac{1}{d}\mu(G_{m})\mu(C)=\frac{1}{d}\delta_\mu(A^*L)\delta_\mu(LA^*).\qedhere
    \]
\end{proof}

There are known examples where the skew product has more than one ergodic measure. In \cref{exampleKarl}, the skew product of $X$ with $\Z/2\Z$ has two orbits, each one being the support of an ergodic measure, and the weighted counting measure is non-ergodic.

\section{On algebraic properties  of the density} 
\label{sec:algebraic}
Let us highlight a corollary which generalizes the fact that, when $\mu$ is a Bernoulli measure with rational values, the density of a rational language is rational~\cite{Berstel1972}. 
\begin{corollary}\label{corollaryExtensionField}
    Let $\K$ be an extension of $\Q$ such that $\mu(L)\in \K\cup\infty$ for
    every rational language $L$. Then, for every rational language $L$ on the alphabet $A$
    which satisfies the hypotheses of \cref{theoremWeightedCountingMeasure}, the
    density of $L$ belongs to $\K$.
\end{corollary}

The hypothesis that $\mu(L)\in \K\cup\infty$ is satisfied when $\mu$ is a Markov measure (or, more generally, a sofic measure) defined by a transition matrix and an initial vector with coefficents in $\K$. 
More details can be found in \cref{sec:markov}. 
It is also satisfied when all the values of $\mu$ on $A^*$ are in $\K$ and the support $X$
of $\mu$ is minimal. Indeed, in this case, consider a rational language $L$. If every word
of $\cL(X)$ is a factor of some word of $L$, then $\delta_\mu(L)>0$ by Equation \eqref{eqDensity},
which implies $\mu(L)=\infty$. Otherwise, the intersection $L\cap \cL(X)$ is
finite and thus the conclusion follows. Let us finish with an example.

\begin{example}\label{ex:parity}
    Let $X$ be the Fibonacci shift (\cref{exampleFibonacci}) and $\A$ be the automaton
    depicted in \cref{figureJClass}. Let $\mu$ be the unique ergodic measure on $X$.
    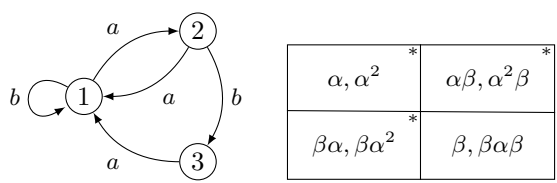
\begin{figure}
      \centering
        \begin{tikzpicture}
            
            \node[state](1)at(180:1){$1$};
            \node[state](2)at(60:1){$2$};
            \node[state](3)at(300:1){$3$};

            \draw[->,label,loop left,left](1) edge node{$b$}(1);
            \draw[->,label,bend left](1)edge node{$a$}(2);
            \draw[->,label,bend left](2)edge node{$a$}(1);
            \draw[->,label,bend left](2)edge node{$b$}(3);
            \draw[->,label,bend left](3)edge node{$a$}(1);
        \end{tikzpicture}
        \quad
        \begin{tikzpicture}
            \matrix (j)  [dclass,minimum width=50pt,minimum height=25pt]{
                $\alpha,\alpha^2$ & $\alpha\beta,\alpha^2\beta$ \\ 
                $\beta\alpha,\beta\alpha^2$ & $\beta,\beta\alpha\beta$ \\
            };
            \drawgrid{2}{2}{j}
            \node[idempotent] at (j-1-1.north east) {*};
            \node[idempotent] at (j-1-2.north east) {*};
            \node[idempotent] at (j-2-1.north east) {*};
        \end{tikzpicture}
        \caption{The automaton $\A$ in \cref{ex:parity} next to the $\JJ$-class of $J_X(M)$, where $M$ is the transition monoid of $\A$. The image of $a$ in $M$ is denoted by $\alpha$ while the image of $b$ is denoted by $\beta$.}\label{figureJClass}
    \end{figure}
    Let $\varphi\from A^*\to M$ be the transition morphism of $\A$ and
    $\alpha=\varphi(a)$, $\beta=\varphi(b)$. Let $R$ be the $\RR$-class of $\alpha$. Note that $\alpha^2$ is an idempotent which belongs to the $\JJ$-class $J_X(M)$. 

    Let $\nu$ be the weighted counting measure on $(R\cup\{0\})\rtimes X$. We claim that $\nu$ is an ergodic measure. First observe that $(R\cup\{0\})\rtimes X$ has two closed stable subsets
    \[
        Y_0 = (\{\alpha\beta,\alpha^2\beta\}\times [b]_X) \cup (\{0\}\times X),\qquad Y_1 = (\{\alpha\beta,\alpha^2\beta\}\times [a]_X) \cup (\{\alpha,\alpha^2\}\times X).
    \]
    Note that $\nu(Y_0)=0$ and  $\nu(Y_1)=1$ (in fact $Y_1$ is the support of $\nu$). Consider the morphism $\psi\from A^*\to\Z/2\Z$ defined by $\psi(a)=1$ and $\psi(b)=0$ and the corresponding skew product $(\Z/2\Z\times X, \tilde T)$. The weighted counting measure $\tilde\nu$ on this skew product is ergodic by \cite[Corollary 8.12]{BertheGouletOuelletNybergBroddaPerrinPetersen2024}. Furthermore, one can show
    that
    \[
        \pi(g,x)=
        \begin{cases}
            (\alpha^2,x) & \mbox{if $g=0$ and $x_{-1}=a$}\\
            (\alpha,x) & \mbox{if $g=1$ and $x_{-1}=a$}\\
            ((\alpha\beta)^2,x) & \mbox{if $g=0$ and $x_{-1}=b$}\\
            (\alpha\beta,x) & \mbox{if $g=1$ and $x_{-1}=b$}
        \end{cases}
    \]
    is a continuous map $\Z/2\Z\times X\to Y_1$ such that $\pi\circ \tilde T = T\circ\pi$ and
    $\tilde\nu\circ\pi^{-1}=\nu$. Therefore the weighted counting measure on $(R\cup\{0\})\rtimes X$
    is ergodic, and   the hypotheses of \cref{theoremWeightedCountingMeasure} are satisfied.

    The values $\mu(w)$ for $w\in A^*$ are all in the quadratic extension $\K=\Q(\lambda)$ where $\lambda$ is the golden ratio \cite[Theorem 2]{Berthe1996}. By \cref{corollaryExtensionField} it follows that $\delta_\mu(\varphi^{-1}(m))\in \K$ for every $m\in J_X(M)$. For instance the language $L=\varphi^{-1}(\alpha)$ satisfies $\mu(LA^*)=\mu(A^*L)=\mu(a)$. Since $\mu(a) = \lambda$ (cf.\ \cite[Example~3.8.19]{DurandPerrin2021}), we get that $\delta_\mu(L)=\frac{1}{2\lambda^2}$ from \cref{theoremWeightedCountingMeasure} (noting that the $\HH$-classes of $J_X(M)$ have size 2).
    Notice also that the language $\psi^{-1}(0)$ has the same intersection with $\cL(X)$ as the submonoid $C^*$ generated by $C=\{aa,aba,b\}$. The language $C^*$ is recognized by $\A$ with $1$ as initial and terminal state. Since $\Z/2\Z\rtimes X$ is uniquely ergodic, we have $\delta_\mu(\psi^{-1}(0))=1/2$.  
\end{example}

\section{Markov and sofic measures}
\label{sec:markov}

This section discusses how our results may be applied to the case of Markov, and more generally sofic measures. The main result, \cref{propositionBerstel}, shows that sofic measures satisfy the condition on field extensions from \cref{corollaryExtensionField}.

First, recall that a \emph{Markov measure} $\mu$ (also called a \emph{Markov chain}) is given by an $A\times A$-stochastic matrix $M$ (its transition matrix) and a stochastic $A$-vector $v$ (its initial vector) which is a left eigenvector of $M$ for the eigenvalue $1$. Then, for $w=a_0a_1\cdots a_{n-1}$, with $a_i\in A$, we define
\[
    \mu(w)=v_{a_0}M_{a_0a_1}\cdots M_{a_{n-2},a_{n-1}}
\]
The measure is invariant because $vM=v$. It is ergodic if the matrix $M$ is irreducible and it is mixing if $M$ is primitive (see \cite{book/Petersen1983}). A shift space $X$ has \emph{finite type} if $\cL(X) = A^*\setminus A^*FA^*$ for some finite set $F\subseteq A^*$, called the set of forbidden blocks. The support of a Markov measure $\mu$ is the shift of finite type defined by the set of forbidden blocks $ab$ where $a,b\in A$ are such that $M_{a,b}=0$. In terms of probability theory, a Markov measure is defined by a sequence of random variables $\zeta_n$, as for a Bernoulli measure, but this time $\zeta_n$ depends on $\zeta_{n-1}$. When $\zeta_n$ depends on $\zeta_{n-k},\ldots,\zeta_{n-1}$ we say that $\mu$ is a Markov measure of \emph{order} $k$.

A \emph{sofic measure} $\nu$ (also called a \emph{hidden Markov chain}) is given by a shift of finite type $X$ on the alphabet $B$, a Markov measure $\mu$ on $X$, and a map $\phi\colon B\to A$ from $B$ onto an alphabet $A$. We extend $\phi$ to a morphism from $B^*$ to $A^*$. Its extension to a map from $A^\Z$ to $B^\Z$ is called a $1$-block map. Then, for $w\in A^*$, we define
\[
    \nu(w)=\mu(\phi^{-1}(w)).
\]
The support of $\nu$ is contained in the sofic shift $Y=\phi(X)$ (a shift $X$ is called \emph{sofic} if $\cL(X)$ is rational). A sofic measure is invariant. It is ergodic if $\mu$ is ergodic (see \cite{BoylePetersen2011}).
\begin{example}\label{exampleBoylePetersen}
    Consider the Markov measure on $B=\{1,2,3\}$ defined by the pair 
    \[
        v=\begin{bmatrix}1/3&1/3&1/3\end{bmatrix},\quad
        M=\begin{bmatrix}0&2/3&1/3\\
        2/3& 1/3&0\\
        1/3&0&2/3\\
        \end{bmatrix}
    \]
    The support of $\mu$ is the shift of finite type $X$ represented in \cref{figureExample} on the left.  Using the $1$-block map $\phi(1)=a$ and $\phi(2)=\phi(3)=b$, we obtain an invariant sofic measure defined by the diagram of \cref{figureExample} on the right.
    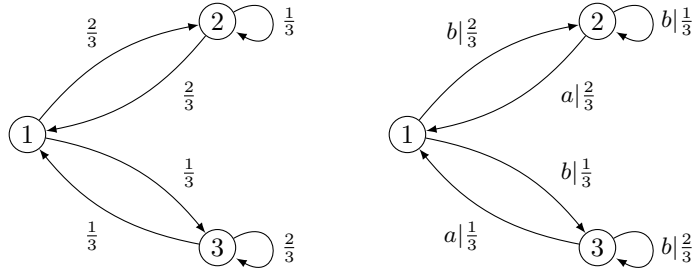
\begin{figure}
        \centering
            \begin{tikzpicture}
                \node[state](1)at(-5,0){$1$};\node[state](2)at(-2.5,1.5){$2$};
                \node[state](3)at(-2.5,-1.5){$3$};

                \draw[->,above,bend left=20,label](1)edge node{$\frac{2}{3}$}(2);
                \draw[->,loop right,right,label](2)edge node{$\frac{1}{3}$}(2);
                \draw[->,above,bend left=20,label,near start,](2)edge node{$\frac{2}{3}$}(1);
                \draw[->,above,bend left=20,near end,label](1)edge node{$\frac{1}{3}$}(3);
                \draw[->,loop right,right,label](3)edge node{$\frac{2}{3}$}(3);
                \draw[->,above,bend left=20,label](3)edge node{$\frac{1}{3}$}(1);
                
                \node[state](1)at(0,0){$1$};
                \node[state](2)at(2.5,1.5){$2$};
                \node[state](3)at(2.5,-1.5){$3$};

                \draw[->,above,bend left=20,label](1)edge node{$b|\frac{2}{3}$}(2);
                \draw[->,loop right,right,label](2)edge node{$b|\frac{1}{3}$}(2);
                \draw[->,above,bend left=20,near start,label](2)edge node{$a|\frac{2}{3}$}(1);
                \draw[->,above,bend left=20,near end,label](1)edge node{$b|\frac{1}{3}$}(3);
                \draw[->,loop right,right,label](3)edge node{$b|\frac{2}{3}$}(3);
                \draw[->,above,bend left=20,label](3)edge node{$a|\frac{1}{3}$}(1);
        \end{tikzpicture}
        \caption{A sofic measure.}\label{figureExample}
    \end{figure}
    It can be shown that this sofic measure is not a Markov measure of order $k$ for any $k$ (see~\cite{BoylePetersen2011}).
\end{example}

A map $\mu\colon A^*\to \R_+$ is \emph{$\R_+$-rational} if there is a morphism $\varphi\colon A^*\to M_n(\R_+)$
from $A^*$ into the monoid of $n\times n$-matrices with coefficients in $\R_+$, a row vector $\lambda$
and a column vector $\gamma$ such that
\[
    \mu(w)=\lambda\mu(w)\gamma,
\]
for every $w\in A^*$. The triple $(\lambda,\varphi,\gamma)$ is called a \emph{linear representation} of $\mu$.
The following is from \cite{Hansel1989} (see also the survey \cite{BoylePetersen2011})

\begin{proposition}\label{propositionHP}
    A probability measure on $A^\Z$ is sofic if and only if the associated probability distribution is $\R_+$-rational.
\end{proposition}

The probability measure will be invariant if and only if the linear representation $(\lambda,\varphi,\gamma)$ can be chosen such that the matrix $P=\sum_{a\in A}\varphi(a)$ is stochastic, with $\lambda P=\lambda$ and $\gamma=\begin{bmatrix}1&1&\ldots&1\end{bmatrix}^t$.
\begin{example}
    A linear representation for the sofic measure of \cref{exampleBoylePetersen} is given by 
    \[
        \varphi(a)=\begin{bmatrix}0&0&0\\2/3&0&0\\1/3&0&0\end{bmatrix},\quad
        \varphi(b)=\begin{bmatrix}0&2/3&1/3\\0&1/3&0\\0&0&2/3\end{bmatrix},
    \]
    with 
    \[
        \lambda=\begin{bmatrix}1/3&1/3&1/3\end{bmatrix},\quad \gamma=\begin{bmatrix}1\\ 1\\ 1\end{bmatrix}.
    \]
\end{example}

Let $L\subseteq A^*$ be a language and $\mu$ be a probability measure on $A^\Z$. The \emph{generating series} of $L$ is the formal series
\[
    f_L(z)=\sum_{n\ge 0}\mu(L\cap A^n)z^n.
\]
Therefore, the density of $L$ (if it exists) is the limit in average of the coefficients of $f_L(z)$.

A formal series $f(z)=\sum_{n\ge 0}f_nz^n$, with $f_n\ge 0$, is $\R_+$-rational if the map $a^n\mapsto f_n$ is $\R_+$-rational. In this case, we have $f(z)=p(z)/q(z)$ for two polynomials $p,q$ with real coefficients. The following statement, originally due to Berstel, is well known (see \cite{Eilenberg1974} for example).
\begin{proposition}\label{propositionBerstel}
    Let $L$ be a rational language. If $\mu$ is a sofic measure, then $f_L(z)$ is $\R_+$-rational. Let $\K$ be a subfield of $\R$ containing the values of $\mu$ on $A^*$. Then $\mu(L)\in \K\cup\{\infty\}$ and $\delta_\mu(L)\in \K$.
\end{proposition}

\begin{proof}
    Let $g\colon A^*\to \R_+$ be defined by 
    \[
        g(w)=\begin{cases}
            \mu(w) & \mbox{if $w\in L$,}\\
            0 & \mbox{otherwise.}
        \end{cases}
    \]
    Thus, $g(w)$ is the product of the values of $\mu$ and of the characteristic function $\chi_L$ of $L$.  The map $w\mapsto\mu(w)$ is $\R_+$-rational by \cref{propositionHP} and $\chi_L$ is $\R_+$-rational since $L$ is rational. Therefore, by \cite[Theorem 5.2]{Eilenberg1974}, the map $g$ is $\R_+$-rational. Set $h_n=\sum_{w\in A^n}g(w)$. Then $h=\sum_{n\ge 0}h_nz^n$ is $\R_+$-rational. Since $h_n=\mu(L\cap A^n)$, this proves that $f_L(z)$ is $\R_+$-rational.

    By \cite[Theorem 7.2]{Eilenberg1974}, there is an integer $p\ge 0$ such that $r_i=\lim_{n\to\infty}h_{np+i}$ exists for $0\le i<p$. Moreover $r_i\in \K$ whenever $h_n\in \K$. This implies that $\delta_\mu(L)=\frac{1}{n}\sum_{0\le i<p}r_i$ exists and is in $\K$. Finally, we have $f_L(z)=p(z)/q(z)$ for two polynomials $p,q\in \R[z]$. Since $h_n\in \K$, we may assume $p,q\in \K[z]$ by \cite[Proposition 3.2]{Eilenberg1974}. Thus, if $\mu(L)=\sum_{n\ge 0}h_n$ is finite, the radius of convergence of $f_L(z)$ is $>1$ and the sum  is equal to $p(1)/q(1)$, which is in $\K$.
\end{proof}

\begin{example}
    Let $\mu$ be the Bernoulli measure on $\{a,b\}^\Z$ defined by $\mu(a)=p$, $\mu(b)=q$.  Let $L=\{a,b\}^*ab$. Then
    \[
        f_L(z)=\sum_{n\ge 0}pqz^{n+2}=\frac{pqz^2}{1-z}
    \]
    and consequently $\delta_\mu(L)=pq$.
\end{example}


\bibliography{probasProfinis.bib}

\appendix
\input{appendix_semigroups}

\end{document}

%% file: appendix_semigroups.tex
\section{Semigroup theory}
\label{appendix-semigroups}

This appendix recalls some basic definitions and results concerning ideals in monoids. For a more detailed exposition, we refer to textbooks such as~\cite{BerstelPerrinReutenauer2009,book/Lallement1979,book/Grillet1995,book/Howie1995,Eilenberg1974,Eilenberg1976}. Most of the section focuses on Green's equivalence relations. These relations are used to produce a partition of a given monoid in terms of the principal ideals generated by its elements. 

A \emph{right ideal} in a monoid $M$ is a set $R$ such that $RM\subseteq R$. A \emph{left ideal} in a monoid $M$ is a set $L$ such that $ML\subseteq L$. A \emph{two-sided ideal} is a set $I$ such that $MIM\subseteq I$. For instance, for $m\in M$, the sets $mM$, $Mm$ and $MmM$ are respectively the smallest right, left and two-sided ideals containing $m$. The quotient of a monoid $M$ by an ideal $I$ is the monoid $M/I=(M\setminus I)\cup\{0\}$ with the operation
\begin{displaymath}
    mn=
    \begin{cases}
        mn & \text{if $mn\notin I$}\\
        0  & \text{otherwise}
    \end{cases}
\end{displaymath}

The Green relations of a monoid $M$ are the equivalence relations defined by
\begin{itemize}
    \item $m\RR n \iff mM=nM$, 
    \item $m\LL n \iff Mm=Mn$,
    \item $m\JJ n \iff MmM=MnM$.
\end{itemize}
In other words, two elements of $M$ are $\RR$-equivalent when they generate the same right ideal, $\LL$-equivalent when they generate the same left ideal, and $\JJ$-equivalent when they generate the same two-sided ideal. One also denotes by $\HH$ the equivalence $\RR\cap \LL$ and by $\DD$ the equivalence $\RR\LL=\LL\RR$. The equivalence classes of these relations are called the $\RR$-classes, $\LL$-classes, $\JJ$-classes, $\HH$-classes and $\DD$-classes respectively. Note that $\HH$ is contained in $\LL$ and $\RR$, which are both contained in $\DD$, which is contained in $\JJ$. 

Consider the quasi-orders defined by $m\leq_{\RR}n$ if $mM\subseteq nM$, and $m\leq_{\LL}n$ if $Mm\subseteq Mn$. A useful property of Green's relations in finite monoids is the following. We say that a monoid $M$ is \emph{stable} if for every $s,t\in M$ such that $s\JJ t$, the following implications hold:
\begin{equation*}
    s\geq_{\LL} t\implies s\LL t,\qquad s\geq_{\RR} t \implies s\RR t.
\end{equation*}
It is well-known that every finite monoid is stable~\cite[Lemma 1.1, Chapter V]{book/Grillet1995}. In stable monoids, $\JJ=\DD$, and thus when dealing with finite monoids we do not need to distinguish between $\JJ$ and $\DD$.

A $\DD$-class $D$ is called \emph{regular} if it contains an idempotent. Every $\HH$-class in $D$ containing an idempotent is a group and the groups corresponding to different $\HH$-classes contained in $D$ are all isomorphic. Moreover every submonoid of $M$ which is a group is contained in a regular $\HH$-class, thus the regular $\HH$-classes are also known as the \emph{maximal subgroups}.

When $M$ is a monoid of partial mappings from a set $Q$ to itself, the Green relations have natural interpretations. Let us adopt the convention that $M$ acts on the right of $Q$ (so the operation on $M$ is reversed composition) and let us define the \emph{kernel} and \emph{image} of an element $m\in M$ as 
\begin{equation*}
    \ker(m) = \{ (p,q)\in Q\times Q \mid pm = qm\},\qquad \Image(m) = \{qm \mid q\in Q\}.
\end{equation*}
Moreover, define the rank of $m$ as $\rank(m) = \Card(\Image(m))$. Then it is not hard to see that the following implications hold:
\[
    m\RR n\implies \ker(m)=\ker(n),\quad m \LL n \implies \Image(m)=\Image(n),\quad m\JJ n \implies \rank(m) = \rank(n).
\]
If $M$ is the monoid of all partial mappings on $Q$, then the reverse implications also hold.

The partition of a monoid using Green's relations is typically depicted using a so-called \emph{eggbox picture}. In an eggbox picture, the $\DD$-classes are represented by boxes where each row is an $\RR$-class and each column an $\LL$-class. The cells (where rows and columns intersect) thus represent the $\HH$-classes. It is customary to indicate the regular $\HH$-classes using an asterisk.

\begin{example}\label{ex:monoid-1}
    Let $M$ be the transition monoid of the automaton $\A$ in \cref{figureMonoidM}. The monoid $M$ has four $\DD$-classes represented in the eggbox picture. Two of them are the class of the identity $1=\varphi(\varepsilon)$ and the class of the empty map  $0=\varphi(a^3)=\varphi(b^3)$. In this example, each $\HH$-class has only one element.
    \begin{figure}
      \centering
      \begin{tabular}{c}
          \begin{tikzpicture}
              \node[state](3)at(0,0){$3$};
              \node[state](1)at(2,0){$1$};
              \node[state](2)at(4,0){$2$};

              \draw[->,above,bend left](1)edge node{$a$}(2);
              \draw[->,above,bend left](2)edge node{$b$}(1);
              \draw[->,above,bend left](1)edge node{$b$}(3);
              \draw[->,above,bend left](3)edge node{$a$}(1);
          \end{tikzpicture}
      \end{tabular}
      \quad
      \begin{tabular}{c}
          \begin{tikzpicture}

              \matrix[dclass] (j3){
                  1\\
              };
              \node[idempotent] at (j3-1-1.north east) {*};

              \matrix[dclass,below right=8pt of j3] (j2) {
                  $\alpha$ & $\alpha\beta$ \\
                  $\beta\alpha$ & $\beta$ \\
              };
          \drawgrid{2}{2}{j2}
              \node[label,anchor=south]at(j2-1-1.north){$1,2$};
              \node[label,anchor=south]at([yshift=-1pt]j2-1-2.north){$1,3$};
              \node[label,anchor=east]at(j2-1-1.west){$1,3$};
              \node[label,anchor=east]at(j2-2-1.west){$1,2$};
              \node[idempotent] at (j2-1-2.north east) {*};
              \node[idempotent] at (j2-2-1.north east) {*};

              \matrix[dclass,below right=8pt of j2,minimum width=25pt] (j1) {
                  $\alpha^2$ & $\alpha^2\beta$ & $\alpha^2\beta^2$ \\
                  $\beta\alpha^2$ & $\beta\alpha^2\beta$ & $\alpha\beta^2$ \\
                  $\beta^2\alpha^2$ & $\beta^2\alpha$ & $\beta^2$ \\
              };
          \drawgrid{3}{3}{j1}
              \node[idempotent] at (j1-1-3.north east) {*};
              \node[idempotent] at (j1-2-2.north east) {*};
              \node[idempotent] at (j1-3-1.north east) {*};
              \node[label,anchor=south]at(j1-1-1.north){$2$};
              \node[label,anchor=south]at(j1-1-2.north){$1$};
              \node[label,anchor=south]at(j1-1-3.north){$3$};
              \node[label,anchor=east]at(j1-1-1.west){$3$};
              \node[label,anchor=east]at(j1-2-1.west){$1$};
              \node[label,anchor=east]at(j1-3-1.west){$2$};

              \matrix[dclass,below right=8pt of j1] (j0){
                  0\\
              };
          \node[idempotent] at (j0-1-1.north east) {*};
      \end{tikzpicture}
  \end{tabular}
\caption{A finite automaton next to the eggbox picture of its transition monoid $M$, where $\alpha$ and $\beta$ denote respectively the image of $a$ and $b$ in $M$.}
    \label{figureMonoidM}
    \end{figure}
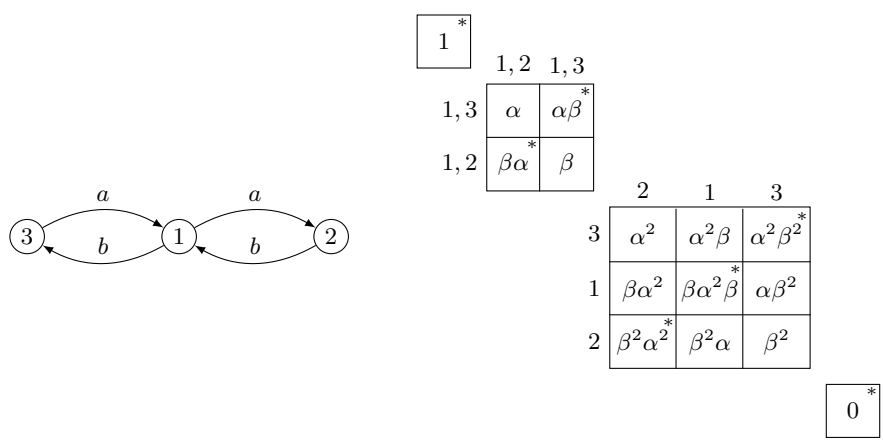
\end{example}

An ideal $I\subseteq M$ is \emph{minimal} if there is no ideal properly contained in it. Every finite monoid $M$ has a unique minimal ideal which is denoted $K(M)$. If $M$ has a zero element $0$, then $K(M)=0$. In that case we say that an ideal $I\ne\{0\}$ is \emph{$0$-minimal} if the only ideal properly contained in $I$ is $\{0\}$. A finite monoid always admits \emph{at least} one $0$-minimal ideal, but it not necessarily unique. In fact a finite monoid with zero admits a unique $0$-minimal ideal whenever it has the following property, known as \emph{primality}: for every $m,n\in M\setminus\{0\}$, there exists $u\in M$ such that $mun\ne 0$. In that case, the unique 0-minimal ideal is composed of a regular $\JJ$-class and zero \cite[Proposition 1.12.9]{BerstelPerrinReutenauer2009}.

\begin{example}\label{ex:monoid-2}
    Continuing with the monoid $M$ of \cref{ex:monoid-1}, the $0$-minimal ideal is the $\DD$-class of maps of rank $1$, which contains for instance $\varphi(a^2)$ and $\varphi(b^2)$. It is formed of $9$ elements which are maps of rank $1$ with range indicated by a label of the corresponding column and domain indicated by a label of the corresponding row. For example, $\varphi(ba^2b)=\varphi(ab^2a)$ (the element in the central cell) is the idempotent with source and range equal to $1$.
\end{example}

We finish by stating a fundamental result which is known as Green's lemma, which can be found in volumes on semigroup theory.
\begin{lemma}[Green's lemma] \label{lem:green}
    Let $r$ and $s$ be $\LL$-equivalent elements of a monoid $M$ and let $u$ and $v$ be such that $ur = s$ and $vs=r$. Then the mappings $x\mapsto ux$ and $y\mapsto vy$ define mutually inverse bijections between the $\RR$-classes of $r$ and $s$. Moreover those bijections preserve $\HH$-classes. 
\end{lemma}